\documentclass[11pt]{article}
\usepackage{a4wide}
\usepackage{amssymb,amsfonts,latexsym,stmaryrd}
\usepackage{amsmath}
\usepackage[PostScript=dvips]{diagrams}
\usepackage{euscript}
\usepackage{QED}

\usepackage{stmaryrd}

\newtheorem{theorem}{Theorem}[section]
\newtheorem{lemma}[theorem]{Lemma}
\newtheorem{definition}[theorem]{Definition}
\newtheorem{example}[theorem]{Example}
\newtheorem{proposition}[theorem]{Proposition}

\newtheorem{corollary}[theorem]{Corollary}

\newcommand{\even}[1]{\mbox{$#1^{\mathsf{even}}$}}

\newcommand{\ident}[1]{\mathsf{id}_{#1}}

\newcommand{\Nat}{{\bf Nat}}

\newcommand{\SET}[1]{\{ #1 \}}
\newcommand{\Fr}{\rightarrow}

\newcommand{\Rest}{{\upharpoonright}}

\newcommand{\La}{\lambda}
\newcommand{\Pfr}{\rightharpoonup}
\newcommand{\Deq}{\approx}
\newcommand{\Over}[1]{\overline{ #1 }}

\newcommand{\Sum}{\sum}

\newcommand{\THEN}{\; \Longrightarrow \;}
\newcommand{\IFF}{\;\mbox{ iff }\;}
\newcommand{\AND}{\; \wedge \;}
\newcommand{\ofcourse}{{!}}

\newcommand{\with}{\, \& \,}

\newcommand{\tensor}{\mbox{$\otimes$}}
\newcommand{\linimpl}{\multimap}
\newcommand{\Games}{\mbox{$\cal{G}$}}
\newcommand{\KG}{\mathbf{K}_{\ofcourse }({\cal G})}
\newcommand{\tunit}{I}

\newlength{\sqpreordheight}
\newlength{\sqpreorddepth}
\newcommand{\Subeq}%
   {\mathbin{\raisebox{-1.02ex}[\sqpreordheight][\sqpreorddepth]%
                            {$\stackrel{\textstyle \sqsubset}{\approx}$}}}

\newcommand{\ie}{\textit{i.e.}~}
\newcommand{\PP}{M^{\circledast}}
\newcommand{\just}{\mathsf{j}}
\newcommand{\pfn}{\Pfr}
\newcommand{\Init}{\mathsf{Init}}
\newcommand{\out}{\mathsf{out}}
\newcommand{\XF}{X^{\flat}}
\newcommand{\Ghf}{\mathcal{G}^{\mathsf{hf}}}

\newcommand{\GG}{\mathcal{G}}
\newcommand{\rarr}{\rightarrow}
\newcommand{\Kb}{\mathbf{K}_{!}}
\newcommand{\der}{\epsilon}
\newcommand{\LL}{\mathcal{L}}
\newcommand{\GL}{\GG_{\LL}}
\newcommand{\GHL}{\Ghf_{\LL}}
\newcommand{\lev}{\mathsf{lev}}
\newcommand{\Tl}{T_{\ell}}
\newcommand{\Tp}{T_{\ell'}}

\newcommand{\Nb}{\Nat^{\flat}_{\bot}}
\newcommand{\noflow}{\not\rarr}
\newcommand{\Lev}{\mathsf{Level}}
\newcommand{\lsem}{\llbracket}
\newcommand{\rsem}{\rrbracket}
\newcommand{\CC}{\mathcal{C}}
\newcommand{\phd}{\phi^{\bullet}}
\newcommand{\psd}{\psi^{\bullet}}
\newcommand{\shr}{\, {!}}
\newcommand{\Osk}{O_{\mathsf{sk}}}
\newcommand{\bind}{\mathbf{bind} \; x = e \; \mathbf{in} \; e'}

\title{Game Semantics for Access Control}

\author{Samson Abramsky and Radha Jagadeesan}

\begin{document}

\maketitle

\begin{abstract}
We introduce a semantic approach to the study of logics for access control and dependency analysis, based on Game Semantics. We use a variant of AJM games with explicit justification (but without pointers). Based on this, we give a simple and intuitive model of the information flow constraints underlying access control. This is used to give strikingly simple proofs of  \emph{non-interference theorems} in robust, semantic versions.
\end{abstract}

\section{Introduction}

In recent years, there has been a significant development of constructive logics and type theories for access control \cite{Abadi06,GargBBPR06,GP06}. The core structure of these logics has turned out to coincide in large part with a calculus previously developed as a basis for various forms of dependency analysis in programming languages \cite{AbadiBHR99}.
This structure can be described quite succinctly as follows. We take a standard type theory as a basic setting. This may be the simply-typed or polymorphic $\lambda$-calculus \cite{Abadi06}, or some form of linear type theory \cite{GargBBPR06}. Such type theories correspond to systems of logic under the Curry-Howard correspondence. We then extend the type theory with a family of monads, indexed by the elements of a ``security lattice'' $\LL$ \cite{Abadi06}.
This lattice can be interpreted in various ways. The basic reading is to think of the elements of the lattice as indicating \emph{security levels}. We shall follow the convention, as in \cite{Abadi06}, that a higher level (more trusted) is \emph{lower down} in the lattice ordering.

The reading which is often followed is to think of an underlying partially ordered set of ``principals'', with the lattice elements corresponding to sets of principals. This leads to the reading of $\Tl A$, where $\Tl$ is the monad indexed at the security level $\ell$, as ``$\ell$ says $A$''.
The monads are type-theoretic counterparts of logical modalities \cite{Kob97}; their use is well established both in logical type theories and in programming languages \cite{Mog91,FairtloughM94,AMPR01}.  We illustrate their use in the specification of access control with an example drawn from Garg and Abadi~\cite{GargA08}.
\begin{example}[~\cite{GargA08}]
Let there be two principals, Bob (a user) and admin (standing for administration).  Let dfile stand for the proposition that a certain file should be deleted. Consider the collection of assertions:
\begin{enumerate}
\item (admin says dfile) $\Rightarrow$ dfile
\item admin says ((Bob says dfile) $\Rightarrow$ dfile )
\item Bob says dfile
\end{enumerate}
Using the unit of the monad with (iii) yields  (admin says (Bob says dfile)).  Using modal consequence with (ii) yields:
\begin{itemize}
\item (admin says (Bob says dfile)) $\Rightarrow$ (admin says dfile)
\end{itemize}
dfile now follows using modus ponens.
\end{example}

The main results which are obtained in this setting, as a basis for access control or dependency analysis, are \emph{non-interference theorems}\footnote{The term `interference' is used in a number of senses in the security literature. Our usage follows that in \cite{Abadi06,GargBBPR06,GP06}.}, which guarantee the \emph{absence} of information flows or logical dependencies which would contradict the constraints expressed by the security lattice $\LL$. A typical example of a non-interference theorem could be expressed informally as follows:
\begin{quote}
No proof of a formula
of the form ``P says $\phi$'' can make any essential use of formulas of the
form ``Q says $\psi$'' unless Q is at the same or  higher security level as P.
In other words, we cannot rely on a lower standard of ``evidence'' or authorization in passing to a higher level.
\end{quote}
In the flow analysis context, it is natural to think of the constraints as ensuring that information does not flow from ``higher'' to ``lower'' variables \cite{AbadiBHR99}. We would then use the same definitions, and obtain the same results --- but with the opposite reading of the security lattice!

Thus far, access control logics have predominantly been studied using proof-theoretic methods.
Our aim in the present paper is to initiate a semantic approach based on \emph{Game Semantics}.

Game Semantics has been developed over the past 15~years as an approach to the semantics of both programming languages and logical type theories \cite{AJ92b,Abramsky2000,Hyland2000,AM97a,Lai97,AHM98,GM04,AGMOS04,Tzev07}. It has yielded numerous full abstraction and full completeness results, in many cases the only such results which have been obtained. There has also been an extensive development of algorithmic methods, with applications to verification \cite{GM00,Abr2002a,MOW05,MW05,LMOW}.

Our aim in the present paper is to show that Game Semantics provides an intuitive and illuminating account of access control, and moreover leads to strikingly simple and robust proofs of interference-freedom.

\paragraph{General Advantages of the Semantic Approach}
Proof-theoretic approaches to negative results such as non-interference properties necessarily proceed by induction over the proof system at hand. This embodies a ``closed world assumption'' that the universe is inductively generated by the syntax of the system, which means that each new system requires a new proof.
A semantic approach, which is carried out in a semantic framework capable of providing models for a wide range of systems, and which supplies \emph{positive reasons} --- structural properties and invariants of the semantic universe --- for the negative results, can be more general and more robust.
We shall give an illustrative example of how semantic non-interference results can be used to obtain results about syntactic calculi in Section~\ref{semsynsec}.

\paragraph{Specific Features of the Games Model}
A number of features of the version of Game semantics which we shall use in this paper are interesting in their own right, and will be developed further in future work.
\begin{itemize}
\item We shall introduce a novel version of AJM games \cite{Abramsky2000} which has a notion of \emph{justifier}, which will be used in the modelling of the access control constraints. This notion plays an important r\^ole in HO-games \cite{Hyland2000}, the other main variant of game semantics, but it assumes a much simpler form in the present setting (in particular, no pointers are needed). This is a first step in the development of a common framework which combines the best features of both styles of game semantics.
\item We also achieve a considerable simplification of the treatment of strategies in AJM games. In particular, we eliminate the need for an ``intensional equivalence'' on strategies \cite{Abramsky2000}.
\end{itemize}

The further contents of the paper are as follows. In Section~2, we develop our variant of AJM games. In Section~3, we describe the games model for access control. We give a semantic treatment of non-interference theorems in Section~4. The relation to syntax is discussed briefly in Section~5. Finally, in Section~6 we conclude with a discussion of directions for future work.

\section{Justified AJM Games}

In this section, we shall introduce a minor variant of AJM games \cite{Abramsky2000} which will provide a basis for our semantics of access control, while yielding a model of Intuitionistic Linear Logic and related languages isomorphic to that given by the usual AJM games.

We wish to refine AJM games by introducing a notion of \emph{justifying move}. A clear intuition for this notion can be given in terms of procedural control-flow. A call of procedure $P$ will have as its justifier the currently active call of the procedure in which $P$ was (statically) declared. Thus the justifier corresponds to the link in the ``static chain'' in compiler terminology for ALGOL-like languages \cite{AG98}. A top-level procedure call will have no justifier --- it will be an ``initial action''. Finally, a procedure return will have the corresponding call as its justifier.

In AJM games, moves are classified as \emph{questions} or \emph{answers}, and
equating procedure calls with questions and returns with answers, we get the appropriate notion of justifier for games.

The notion of justifier plays a central r\^ole in Hyland-Ong (HO) games \cite{Hyland2000}. In that context, the identification of the justifier of a move in a given play involves an additional structure of ``justification pointers'', which are a considerable complication. The need for this additional structure arises for two reasons:
\begin{itemize}
\item Firstly, the treatment of \emph{copying} in HO games allows multiple occurrences of the same move  in a given play. This means that extra structure is required to identify the ``threads'' corresponding to the different copies. By contrast, plays in AJM games are naturally \emph{linear}, \ie moves only occur once, with the threads for different copies indicated explicitly.
\item The other source of the need for explicit indication of justifiers in plays is that the justification or enabling relation is in general not functional in HO games; a given move may have several possible justifiers, and we must indicate explicitly which one applies. In fact, in the original version of HO games \cite{Hyland2000}, justification \emph{was} functional; the relaxation to more general enabling relations was introduced later for convenience \cite{McC98}, given that in the HO format, justification pointers were going to be used anyway.

It turns out that unique justifiers can be defined straightforwardly in AJM games; the only change which is required is a minor one to the definition of  linear implication. This follows the device used in \cite{Hyland2000} to preserve the functionality of justification.

\end{itemize}

\noindent Given that we have both linearity of plays and unique justifiers, we get a very simple, purely ``static'' notion of justification which is determined by the game, and requires no additional information at the level of plays. The resulting notion of AJM games is equivalent to the standard one as a model of ILL, but carries the additional structure needed to support our semantics for access control.

We shall now proceed to describe the category of justified AJM games. Since a detailed account of the standard AJM category can be found in \cite{Abramsky2000} and the differences are quite minor, we shall only provide a brief outline, emphasizing the points where something new happens.

\subsection{The Games}

A game is a structure $A = (M_A,\lambda_A, \just_{A}, P_A,{\Deq}_A)$, where
\begin{itemize}
\item $M_A$ is the set of moves.
\item $\lambda_A: M_A \rightarrow \{ P,O \} \times \SET{Q,A}$ is the
  labelling function.

  The labelling function indicates if a move is by Player (P) or
  Opponent (O), and if a move is a question (Q) or an answer (A). The
  idea is that questions correspond to requests for data or procedure calls, while
  answers correspond to data
  ({\em e.g.} integer or boolean values) or procedure returns. In a higher-order
  context, where arguments may be functions which may themselves be
  applied to arguments, all four combinations of Player/Opponent with
  Question/Answer are possible. Note that $\lambda_A$ can be decomposed into
  two functions $\lambda_{A}^{PO}:M_A\rightarrow\{P,O\}$ and
  $\lambda_{A}^{QA}:M_A\rightarrow\{Q,A\}$.

We write
\[  \{ P, O \} \times \{ Q, A \} = \{ PQ, PA, OQ, OA \} \]
\[ M_A^P =  \lambda_A^{-1}(\{P \}\times\SET{Q,A}), \quad
  M_A^O =
  \lambda_A^{-1}(\{O \}\times\SET{Q,A}) \]
\[ M_A^Q =
  \lambda_A^{-1}(\{P,O \}\times\SET{Q}), \quad
  M_A^A =
  \lambda_A^{-1}(\{P,O \}\times\SET{A})
\]
and define
\[  \overline{P} = O, \;\; \overline{O} = P, \]
\[  \overline{\La_{A}^{PO}}(a)= \overline{\La_{A}^{PO}(a)}, \;\;
  \Over{\La_{A}}= \langle \Over{\La_{A}^{PO}},\La_{A}^{QA} \rangle . \]

\item The justification function $\just_{A} : M_{A} \pfn M_{A}$ is a partial function on moves satisfying the following conditions:
\begin{itemize}
\item For each move $m$, for some $k>0$, $\just_{A}^{k}(m)$ is undefined, so that the forest of justifiers is well-founded. A move $m$ such that $\just_{A}(m)$ is undefined is called \emph{initial}; we write $\Init_{A}$ for the set of initial moves of $A$.
\item $P$-moves must be justified by $O$-moves, and vice versa; answers must be justified by questions.
\end{itemize}

\item Let $M_A^{\circledast}$ be the set of all finite sequences $s$ of
  moves satisfying the following conditions:
\begin{description}
\item[\textbf{(p1)} Opponent starts] If $s$ is non-empty, it starts with an O-move.
\item[\textbf{(p2)} Alternation] Moves in $s$ alternate between O and P.
\item[(p3) Linearity] Any move occurs at most once in $s$.
\item[(p4) Well-bracketing] Write each answer $a$ as $)_{a}$ and the corresponding question $q = \just_{A}(a)$ as $(_{a}$. Define the set $W$ of \emph{well-bracketed strings} over $A$ inductively as follows: $\varepsilon \in W$; $u \in W \; \Rightarrow \; (_{a} \, u  \, )_{a} \in W$; $u, v \in W \; \Rightarrow \; uv \in W$. Then we require that $s$ is a prefix of a string in $W$.
\item[(p5) Justification] If $m$ occurs in $s$, $s = s_{1}ms_{2}$, then the justifier $\just_{A}(m)$ must occur in $s_{1}$.
\end{description}

Then $P_A$, the set of  \emph{positions} of the game, is a
non-empty prefix-closed subset of $M_A^{\circledast}$.

The conditions {\bf (p1)}--{\bf (p5)} are  global rules
applying to all games.

\item
$\Deq_{A}$ is an equivalence relation on $P_{A}$ satisfying
\[ \begin{array}{ll}
\mbox{{\bf (e1)}} &  s \Deq_{A} t \THEN\ \lambda_{A}^{\star}(s) = \lambda_{A}^{\star}(t) \\
\mbox{{\bf (e2)}} &  s \Deq_{A} t, s' \sqsubseteq s, t' \sqsubseteq t,
|s'| = |t'|
\THEN  s' \Deq_{A} t' \\
\mbox{{\bf (e3)}} &  s \Deq_{A} t, sa \in P_{A} \THEN \exists b. \, sa \Deq_{A} tb .
\end{array} \]
Here $\lambda_{A}^{\star}$ is the extension of $\lambda_{A}$ to sequences; while $\sqsubseteq$ is the prefix ordering.
Note in particular that {\bf (e1)} implies that if $s \Deq_{A} t$, then $|s| = |t|$.
\end{itemize}

\noindent If we compare this definition to that of standard AJM games, the new component is the justification function $\just_{A}$. This allows a simpler statement of the well-bracketing condition.
The new conditions on plays are Linearity and Justification. These hold automatically for the interpretation of ILL types in AJM games as given in \cite{Abramsky2000}.

\subsection{Constructions}

We now describe the constructions on justified AJM games corresponding to the ILL connectives. These are all defined exactly as for the standard AJM games in \cite{Abramsky2000}, with justification carried along as a passenger and defined in the obvious componentwise fashion, with the sole exception of linear implication.
\begin{description}
\item[Times]
We define the tensor product $A \tensor B$ as follows.
\begin{itemize}
\item $M_{A \tensor B} = M_A + M_B$, the disjoint union of the two move
sets.
\item $\lambda_{A \tensor B} = [\lambda_A,\lambda_B]$, the source tupling.
\item $\just_{A \tensor B} = \just_{A} + \just_{B}$.
\item
$P_{A \tensor B} = \{ s \in \PP_{A \tensor B} \mid s \Rest A \in P_{A} \; \wedge \; s \Rest B \in P_{B} \}$.
\item $s \Deq_{A \tensor B} t \IFF  s \Rest A \Deq_{A} t \Rest_{A} \;
\wedge \; s \Rest B \Deq_{B} t \Rest B \; \wedge \; \out^{\star}(s) =
\out^{\star}(t).$
\end{itemize}
Here $\out : \Sigma_{i \in I} X_{i} \rightarrow I :: (x \in X_{i}) \mapsto i$ maps an element of a disjoint union to the index of its summand. Concretely in this case,  $\out(m)$ is 1 if $m \in M_{A}$, and 2 if $m \in M_{B}$.
Note that there is no need to formulate a `stack condition' explicitly as in \cite{Abramsky2000}, since this is implied by the component-wise definition of the justification function.
\item[Tensor Unit ]
The tensor unit is given by
\[ \tunit = (\varnothing,\varnothing, \varnothing, \{\epsilon \}, \{ (\epsilon , \epsilon )
\} ). \]
\item[Additive Conjunction]
The game $A \with B$ is defined as follows.
\begin{eqnarray*}
M_{A \with B} &=& M_A + M_B  \\
\lambda_{A \with B} &=& [\lambda_A,\lambda_B] \\
\just_{A \tensor B} & = & \just_{A} + \just_{B}\\
P_{A \with B}&=&  P_A + P_B  \\
\Deq_{A \with B}&=& \Deq_{A} + \Deq_{B} .
\end{eqnarray*}

\item[Bang]
The game $\ofcourse A$ is defined as the ``infinite symmetric tensor power''
of $A$.
The symmetry is built in via the equivalence relation on positions.

\begin{itemize}
\item $M_{\ofcourse A} =\omega\times M_A=\Sum_{i\in\omega} M_A $, the
  disjoint union of countably many copies of the moves of $A$.  So, moves
  of $\ofcourse A$ have the form $(i, m )$, where $i$ is a natural
  number, called the index, and $m$ is a move of $A$.
\item Labelling is by source tupling:
$$ \lambda_{\ofcourse A}(i, a) =  \lambda_A (a) . $$
\item Justification is componentwise: $\just_{\ofcourse A}(i, m) = (i, \just_{A}(m))$.
\item We write $s\Rest i$ to indicate the restriction to moves with
  index $i$.
  \[ P_{\ofcourse A}  = \{ s \in \PP_{\ofcourse A} \mid (\forall i \in \omega) \; s\Rest i \in P_A \} \, . \]

\item Let $S(\omega )$ be the set of permutations on $\omega$.
\[ s \Deq_{\ofcourse A} t \; \Longleftrightarrow \; ( \exists \pi \in
S(\omega ) )
[(\forall i \in \omega . \, s \Rest i \Deq_{A} t \Rest \pi (i) ) \;
\wedge \; (\pi \, \circ \, {\tt fst})^{\ast}(s) = {\tt fst}^{\ast}(t) ] .
\]

\end{itemize}

\item[Linear Implication]
The only subtlety arises in this case.
The justifier of an initial move in $A$ played within $A \linimpl B$ should be an initial move in $B$; but which one? To render this unambiguous, we make a disjoint copy of the moves in $A$ for each initial move in $B$. A similar device is used in \cite{Hyland2000}. We write  $\Sigma_{b \in \Init_{B}} M_A$ for this disjoint union of copies of $M_{A}$, which is equivalently defined as $\Init_{B} \times M_{A}$.
\begin{itemize}
\item $M_{A \linimpl B} = (\Sigma_{b \in \Init_{B}}M_A) + M_B$.
\item $\lambda_{A \linimpl B} = [[\overline{\lambda_{A}} \mid b \in \Init_{B}],\lambda_B]$.
\item We define justification by cases. We write $m_{b}$, for $m \in M_{A}$ and $b \in \Init_{B}$, for the $b$-th copy of $m$.
\[ \begin{array}{lcl}
\just_{A \linimpl B}(m_{b}) & = & \left\{ \begin{array}{ll}
b, & m \in \Init_{A} \\
(\just_{A}(m))_{b}, & m \not\in \Init_{A}
\end{array} \right. \\
\just_{A \linimpl B}(m) & = & \just_{B}(m), \qquad \quad \, m \in M_{B} \, .
\end{array}
\]
\item We write $s \Rest A$ to indicate the restriction to moves in  $\Sigma_{b \in \Init_{B}} M_{A}$, replacing each $m_{b}$ by $m$.
\[ P_{A \linimpl B} = \{ s \in \PP_{A \linimpl B} \mid s \Rest A \in P_{A} \; \wedge \; s \Rest B \in P_{B} \} \]
Note that Linearity for $A$ implies that \emph{only one copy} $m_{b}$ of each $m \in M_{A}$ can occur in any play $s \in P_{A \linimpl B}$.

\item $s \Deq_{A \linimpl B} t \IFF  (\forall b \in \Init_{B}) \, s \Rest A,b \Deq_{A} t \Rest A,b \;
\wedge \; s \Rest B \Deq_{B} t \Rest B \; \wedge \; \out^{\star}(s) =
\out^{\star}(t).$
\end{itemize}
Note that, by {\bf (p1)}, the first move in any position in
$P_{A \linimpl B}$ must be in $B$.

\end{description}

\paragraph{Basic Types} Given a set $X$, we define the \emph{flat game} $X^{\flat}$ over $X$ as follows:
\begin{itemize}
\item $M_{\XF} = \{ q_{0} \} + X$
\item $\lambda_{\XF}(q_{0}) = OQ$, $\lambda_{\XF}(x) = PA$ for $x \in X$.
\item $j_{\XF}(q_{0})$ is undefined (so $q_{0}$ is initial); $j_{\XF}(x) = q_{0}$.
\item Plays in $\XF$ are prefixes of sequences $q_{0}x$, $x \in X$.
\item The equivalence $\Deq_{\XF}$ is the identity relation.
\end{itemize}
For example, we obtain a game $\Nat^{\flat}$ for the natural numbers.

\subsection{Strategies}

Our aim is to present a reformulation of strategies for AJM games, which is equivalent to the standard account in \cite{Abramsky2000}, but offers several advantages:
\begin{itemize}
\item A major drawback of AJM games is that strategies must be quotiented by an equivalence to obtain a category with the required structure. This is workable, but lacks elegance and impedes intuition. \emph{This problem is completely eliminated here}: we present a notion of strategy which is `on the nose', without any quotient.
\item At the same time, the existing notions are related to the new approach, and the standard methods of defining AJM strategies can still be used.
\item Although we shall not elaborate on this here, the order-enriched structure of strategies is vastly simplified, since the ordering is now simply subset inclusion.
\end{itemize}
The idea of using strategies saturated under the equivalence relation on plays can be found in \cite{BDER97}; but that paper concerned a `relaxed' model, which could be used for Classical Linear Logic, and did not establish any relationship with the standard AJM notions.

\begin{definition}
A \emph{strategy} on a game $A$ is a non-empty set $\sigma \subseteq \even{P_A}$ of even-length plays satisfying the following conditions:
\begin{description}
\item[Causal Consistency]  $sab \in \sigma \THEN s \in \sigma$
\item[Representation Independence] $s \in \sigma \; \wedge \; s \Deq_{A} t \THEN t \in \sigma$
\item[Determinacy] $sab, ta'b' \in \sigma \; \wedge \; sa \Deq_{A} ta'  \THEN sab \Deq_{A} ta'b'$.
\end{description}

\end{definition}

\noindent To relate this to the usual notion of AJM strategies, we introduce the notion of \emph{skeleton}.
\begin{definition}
A \emph{skeleton} of a strategy $\sigma$ is a non-empty causally consistent subset $\phi \subseteq \sigma$ which  satisfies the following condition:
\begin{description}
\item[Uniformization] $\forall sab \in \sigma . \, s \in \phi \THEN \exists ! b' . \, sab' \in \phi$.
\end{description}
\end{definition}

\noindent Note that the play $sab'$ whose existence is asserted by Uniformization satisfies: $sab \Deq_{A} sab'$. This follows immediately from $\phi \subseteq \sigma$ and Determinacy.

\begin{proposition}
Let $\phi$ be a skeleton of a strategy $\sigma$. Then $\phi$ satisfies the following properties:
\begin{itemize}
\item \textbf{Functional Determinacy}: $sab, sac \in \sigma \THEN b=c$
\item \textbf{Functional Representation Independence}:
\[ sab \in \phi \AND t \in \phi \AND sa \Deq_{A} ta' \THEN \exists ! b' . \, ta'b' \in \phi \AND sab \Deq_{A} ta'b'  . \]
\end{itemize}
\end{proposition}

\begin{proof}
\begin{itemize}
\item Functional Determinacy. If $sab, sac \in \phi$, then $sab \in \sigma$ and $s \in \phi$. By Uniformization, the $b$ such that $sab$ is in $\phi$ is unique, so $b = c$.

\item Functional Representation Independence. Suppose that $sab \in \phi$, $t \in \phi$, $sa \Deq_{A} ta'$. Then $sab \in \sigma$, and by \textbf{(e3)}, for some $b''$, $sab \Deq_{A} ta'b''$. By Representation Independence, $ta'b'' \in \sigma$. By Uniformization, for some unique $b'$, $ta'b' \in \phi$, and by the remark after the definition of skeleton, $sab \Deq_{A} ta'b'$.
\end{itemize}
Hence we conclude.
\end{proof}

\begin{definition}
We shall define a skeleton more generally --- independently of any strategy --- to be a non-empty, causally consistent set of even-length plays, satisfying Functional Determinacy and Functional Representation Independence.
\end{definition}
Note that a skeleton $\phi$ is exactly the usual notion of AJM strategy such that $\phi \Subeq \phi$ \cite{Abramsky2000}.

Given a skeleton $\phi$, we define $\phd = \{ t \mid \exists s \in \phi . \, s \Deq_{A} t \}$.

\begin{proposition}
If $\phi$ is a skeleton, then $\phd$ is a strategy, and $\phi$ is a skeleton of $\phd$.
\end{proposition}
\begin{proof}
We verify the conditions for $\phd$ to be a strategy.
\begin{itemize}
\item Causal Consistency. If $ta'b' \in \phd$, then for some $sab \in \phi$, $sab \Deq_{A} ta'b'$. Since $\phi$ is causally consistent, $s \in \phi$, and $s \Deq_{A} t$, hence $t \in \phd$.

\item Representation Independence. If $t \Deq_{A} s \in \phi$ and $u \Deq_{a} t$, then $u \Deq_{A} s$, and hence $u \in \phd$.

\item Determinacy. Suppose $sab, ta'b' \in \phd$. This means that $sab \Deq_{A} s_{1}a_{1}b_{1} \in \phi$, and $ta'b' \Deq_{A} s_{2}a_{2}b_{2} \in \phi$. Now $sa \Deq_{A} ta'$ implies that $s_{1}a_{1} \Deq_{A} s_{2}a_{2}$. By Functional Determinacy and Functional Representation Independence, $s_{1}a_{1}b_{1} \Deq_{A} s_{2}a_{2}b_{2}$.
Hence $sab \Deq_{A} ta'b'$, as required.
\end{itemize}

\noindent Now we verify that $\phi$ is a skeleton of $\phd$, \ie that Uniformization holds.
Suppose that $sab \in \phd$ and $s \in \phi$. By Functional Representation Independence, there is a unique $b'$ such that $sab' \in \phi$ and $sab \Deq_{A} sab'$. By Functional Determinacy, this is the unique $b'$ such that $sab' \in \phi$.
\end{proof}

\begin{proposition}
\label{skeldotprop}
If $\phi$ is a skeleton of $\sigma$, then $\phd = \sigma$.
\end{proposition}
\begin{proof}
Certainly $\phd \subseteq \sigma$ by Representation Independence.
We prove the converse by induction on the length of $s \in \sigma$. The basis case for $\varepsilon$ is immediate. For $sab \in \sigma$, by induction hypothesis $s \in \phd$. Hence for some $s' \in \phi$, $s \Deq_{A} s'$. We can find $a'$, $b'$ such that $sab \Deq_{A} s'a'b'$. By Representation Independence, $s'a'b' \in \sigma$. By Uniformization, for some $b''$, $s'a'b'' \in \phi$, where $s'a'b' \Deq_{A} s'a'b''$. Hence $sab \in \phd$.
\end{proof}

\begin{corollary}
\label{dotcorr}
$\phi$ is a skeleton of $\sigma$ if and only if $\phi$ is a skeleton, and $\phd = \sigma$.
\end{corollary}

\begin{proposition}
\label{existskelprop}
Every strategy $\sigma$ has a skeleton $\phi$.
\end{proposition}
\begin{proof}
We define a family of sets of plays $\phi_{k}$ by induction on $k$. $\phi_{0} = \{ \varepsilon \}$.
To define $\phi_{k+1}$, for each $s \in \phi_{k}$ and $a \in M_{A}^{O}$, consider the set $X_{s,a}$ of all plays in $\sigma$ of the form $sab$ for some $b$. Note that $(s, a) \neq (s', a')$ implies that $X_{s,a} \cap X_{s', a'} = \varnothing$. Let $C$ be the family of all non-empty $X_{s,a}$. Then $\phi_{k+1}$ is a choice set for $C$ which selects exactly one element of each member of $C$.
Finally, define $\phi = \bigcup_{k \in \omega}  \phi_{k}$.
It is immediate from the construction that $\phi$ is a skeleton for $\sigma$.
\end{proof}

We now recall the definition of the preorder on skeletons (standard AJM strategies) from \cite{Abramsky2000}:
\[ \phi\Subeq\psi \; \equiv \;
sab\in\phi, s'\in \psi, sa
  \Deq s'a' \THEN \exists b'.\ [s'a'b'\in\psi \AND sab\Deq s'a'b'].
\]

\begin{proposition}
\label{preordprop}
For skeletons $\phi$, $\psi$ on $A$:
$\phi \Subeq \psi\IFF \phd \subseteq \psd$.
\end{proposition}
\begin{proof}
Suppose firstly that $\phd \subseteq \psd$, and that $sab \in \phi$, $t \in \psi$, $sa \Deq_{A} ta'$.
Then $sab \in \psd$, so there is some $s_{1}a_{1}b_{1} \in \psi$ with $s_{1}a_{1}b_{1} \Deq_{A} sab$.
Then $s_{1}a_{1} \Deq_{A} ta'$, so by Functional Representation Independence, there exists a unique $b'$ such that $ta'b' \in \psi$, and $ta'b' \Deq_{A} s_{1}a_{1}b_{1} \Deq_{A} sab$. Thus $\phi \Subeq \psi$, as required.

For the converse, we assume that $\phi \Subeq \psi$. It is sufficient to prove that $\phi \subseteq \psd$, which we do by induction on the length of plays in $\phi$. The base case is immediate. Now suppose that $sab \in \phi$. By induction hypothesis, $s \in \psd$. Then for some $s' \in \psi$, $s \Deq_{A} s'$. For some $a'$, $s'a' \Deq_{A} sa$. Since $\phi \Subeq \psi$, there exists $b'$ such that $s'a'b'\in\psi \AND sab\Deq s'a'b'$. But then $sab \in \psd$, as required.
\end{proof}

We can now obtain a rather clear picture of the relationship between partial equivalence classes of strategies as in \cite{Abramsky2000}, and strategies and their skeletons in the present formulation.

\begin{proposition}
For  any strategy $\sigma$:
\begin{enumerate}
\item Any two skeletons of  $\sigma$ are equivalent.
\item $\sigma = \bigcup \{ \phi \mid \mbox{$\phi$ is a skeleton on $\sigma$} \}$.
\item For any skeleton $\phi$ of $\sigma$:
$\sigma = \bigcup \{ \psi \mid \psi \Deq \phi \}$.
\end{enumerate}

\end{proposition}
\begin{proof}
1. If $\phi$, $\psi$ are skeletons of $\sigma$, by Proposition~\ref{skeldotprop}  $\phd = \psd$, hence by Proposition~\ref{preordprop}, $\phi \Subeq \psi$ and $\psi \Subeq \phi$.

2. The right-to-left inclusion is clear. For the converse, given any $s \in \sigma$, we can guide the choices made in the construction of $\phi$ in the proof of Proposition~\ref{existskelprop} to ensure that $s \in \phi$.

3. The right-to-left inclusion follows from Proposition~\ref{preordprop}. The converse follows from part 2, Corollary~\ref{dotcorr}, and Proposition~\ref{preordprop}.
\end{proof}

Finally, we show how \emph{history-free strategies}, which play an important r\^ole in AJM game semantics, fit into the new scheme.
We define a strategy to be \emph{history-free} if it has a skeleton $\phi$ satisfying the following additional conditions:
\begin{itemize}
\item $sab, tac \in \phi \THEN b=c$
\item $sab, t\in\phi, ta\in P_A \THEN tab \in \phi$.
\end{itemize}
As in \cite{Abramsky2000}, a skeleton is history-free if and only if it is generated by a function on moves.

\paragraph{Constructions on strategies}
In the light of these results, we can define a strategy $\sigma$ by defining a skeleton $\phi$ and then taking $\sigma = \phd$. In particular, this is the evident method for defining history-free strategies. Thus all the constructions of particular strategies carry over directly from \cite{Abramsky2000}.

What of operations on strategies, such as composition, tensor product etc.? We can define an operation $O$ on strategies via an operation $\Osk$ on skeletons as follows. Given strategies $\sigma$, $\tau$, we take skeletons $\phi$ of $\sigma$ and $\psi$ of $\tau$, and define
\[ O(\sigma, \tau) = \Osk(\phi,\psi)^{\bullet} . \]
Of course, this definition should be independent of the choice of skeletons.

\begin{proposition}
An operation $O$ `defined' in terms of $\Osk$ as above is well-defined and monotone with respect to subset inclusion if and only if $\Osk$ is monotone with respect to $\Subeq$.
\end{proposition}
\begin{proof}
Suppose that $\Osk$ is monotone with respect to $\Subeq$ and that we are given $\sigma \subseteq \sigma'$, $\tau \subseteq \tau'$, and skeletons $\phi$ of $\sigma$, $\phi'$ of $\sigma'$, $\psi$ of $\tau$, $\psi'$ of $\tau'$. By Proposition~\ref{preordprop},  $\phi \Subeq \phi'$ and $\psi \Subeq \psi'$. Then $\Osk(\phi, \psi) \Subeq \Osk(\phi',\psi')$, and so, using Proposition~\ref{preordprop} again,
\[ O(\sigma, \tau) = \Osk(\phi, \psi)^{\bullet} \subseteq \Osk(\phi', \psi')^{\bullet} = O(\sigma', \tau') . \]
Note that, taking $\sigma = \sigma'$ and $\tau = \tau'$, this also shows that $O$ is well-defined.

\noindent Conversely, suppose that $O$ is well-defined and monotone, and consider $\phi \Subeq \phi'$, $\psi \Subeq \psi'$. Let $\sigma = \phd$, $\sigma'=\phi'^{\bullet}$, $\tau = \psd$, $\tau'=\psi'^{\bullet}$. By Proposition~\ref{preordprop}, $\sigma \subseteq \sigma'$ and $\tau \subseteq \tau'$.
Then
\[ \Osk(\phi, \psi)^{\bullet} = O(\sigma, \tau)   \subseteq O(\sigma', \tau') = \Osk(\phi', \psi')^{\bullet}    \]
so by Proposition~\ref{preordprop}, $\Osk(\phi, \psi) \Subeq \Osk(\phi',\psi')$.
\end{proof}

Thus again, all the operations on strategies from \cite{Abramsky2000} carry over to the new scheme.

\subsection{Categories of Games}
We build a category \Games:
\begin{eqnarray*}
\mbox{Objects} &:& \mbox{Justified AJM Games} \\
\mbox{Morphisms} &:& \mbox{$\sigma:A\Fr B$ $\equiv$ strategies on $A \linimpl B$.}
\end{eqnarray*}

\paragraph{Copy-Cat}
The \emph{copy-cat strategy} \cite{AJ92b}
is defined by:
\[ \ident{A} = \{ s \in P^{\tt even}_{A \linimpl A} \mid \; s
{\upharpoonright} 1 \Deq_{A} s {\upharpoonright} 2 \} . \]

\paragraph{Composition}
We  need a slight modification of the definition of composition as given in \cite{AJ92b,Abramsky2000}, to fit the revised definition of linear implication.
Suppose we are given $\sigma: A
\rightarrow B$ and
$\tau: B \rightarrow C$. We define:
\[ \begin{array}{ccl}
\sigma \| \tau & = & \{ s \in M_{(A \linimpl B) \linimpl C}^{\star} \mid
s \Rest A, B \in \sigma \; \wedge \; s \Rest B, C \in \tau \} \\
\sigma; \tau & = & \{ s {\upharpoonright} A,C \mid \; s \in \sigma \| \tau
 \} \, .
\end{array} \]
Note that
\[ \begin{array}{lcl}
M_{(A \linimpl B) \linimpl C} & = & (\Sigma_{c \in \Init_{C}} (\Sigma_{b \in \Init_{B}} M_{A}) + M_{B}) + M_{C} \\
& \cong & \Sigma_{c \in \Init_{C}} (\Sigma_{b \in \Init_{B}} M_{A}) \; + \; \Sigma_{c \in \Init_{C}} M_{B}  \; + \;  M_{C}
\end{array}
\]
so we can regard $s\Rest B, C$ as a play in $B \linimpl C$. Similarly, $s \Rest A, B$, where we erase all tags from $C$, can be regarded as a play in $A \linimpl B$. Finally, $s \Rest A, C$, where all tags from $B$ are erased, so that $(m_{b})_{c}$ is replaced by $m_{c}$, can be regarded as a play in $A \linimpl C$. This last transformation, erasing tags in $B$, corresponds to the elision of justification pointers in the definition of composition for HO games \cite{Hyland2000}.

\begin{proposition}
\Games\ equipped with composition of strategies, and with the copy-cat strategies $\ident{A} : A \rightarrow A$ as identities, is a category. There is also a sub-category $\Ghf$ of history-free strategies.
\end{proposition}

\noindent The further development of these categories as models of ILL proceeds exactly as in \cite{Abramsky2000}.
The main point to note is that monoidal closure still works, in exactly the same way as in \cite{Abramsky2000}.
Indeed, we have an isomorphism
\[ (A \otimes B) \linimpl C \cong A \linimpl (B \linimpl C) \]
induced concretely by the bijection on moves
\[ \begin{array}{lcl}
M_{(A \otimes B) \linimpl \, C} & = & \Sigma_{c \in \Init_{C}} (M_{A} + M_{B}) + M_{C} \\
& \cong & \Sigma_{c \in \Init_{C}} M_{A} + (\Sigma_{c \in \Init_{C}} M_{B} + M_{C}) \\
& = & M_{A \linimpl \, (B \linimpl \, C)}
\end{array}
\]
since the initial moves in $B \linimpl C$ are just those in $C$. Since arrows are defined as strategies on the internal hom, this immediately yields the required adjunction
\[ \GG(A \otimes B, C) \cong \GG(A, B \linimpl C) . \]
The monoidal structure for $\otimes$ is witnessed similarly by copy-cat strategies induced by bijections on move sets, as in \cite{AJ92b}. Thus we get a symmetric monoidal closed (SMC) category $(\GG, {\otimes}, I, {\linimpl})$, with an SMC sub-category $\Ghf$.

Next we note that there are natural transformations
\[ \der_{A} : \ofcourse A \rightarrow A, \qquad \delta_{A} : \, !A \rightarrow \, !!A \]
and a functorial action for $!$ which endow it with the structure of a comonad. The counit $\der$ plays copycat between $A$ and one fixed index in $!A$, while the comultiplication $\delta$ uses a pairing function
\[ p : \omega \times \omega \rightarrowtail \omega  \]
to play copycat between pairs of indices in $!!A$ and indices in $!A$. The functorial action $!\sigma : \, !A \rightarrow \, !B$ simply plays $\sigma$ componentwise in each index. The coding-dependence in these constructions is factored out by the equivalence $\Deq$.

The co-Kleisli category $\Kb(\GG)$ for this comonad has arrows $!A \rarr B$, with identities given by the counits $\der_{A}$. Composition is defined via \emph{promotion}: given $\sigma : !A \rarr B$, we define
\[ \sigma^{\dagger} = \delta_{A} ; !\sigma : \, !A \rarr \, !B . \]
The Kleisli composition of $\sigma$ with $\tau : !B \rarr C$ is then $\sigma^{\dagger} ; \tau : \, !A \rarr C$.

The additive conjunction is the product in the coKleisli category, while $I$ is the terminal object. There are  \emph{exponential isomorphisms}
\[ !(A \with B) \cong \; !A \, \otimes \, !B, \qquad !I \cong I. \]
This ensures that the coKleisli category is cartesian closed: defining $A \Rightarrow B = \, !A \linimpl B$,
we have
\[ \begin{array}{lcl}
\KG (A \with B, C) & = & \Games (\ofcourse (A \with B), \, C) \\
& \cong & \Games (\ofcourse A \, \tensor\,  \ofcourse B , \, C) \\
& \cong & \Games (\ofcourse A , \, \ofcourse B \linimpl C) \\
& = & \KG (A, \, B \Rightarrow C).
\end{array} \]
Thus we have a symmetric monoidal closed category $(\GG, {\otimes}, I, {\linimpl})$ and a cartesian closed category $(\KG, {\with}, I, {\Rightarrow})$. There is, automatically, an adjunction between $\GG$ and its coKleisli category $\KG$. This adjunction is moreover monoidal by virtue of the exponential isomorphisms. This provides exactly the required structure for a model of ILL \cite{Bie95,MMPR}. Moreover, all this structure cuts down to the history-free strategies.  The interpretation of the ILL type theory lives inside the history-free sub-category $\Ghf$.

\section{The Model}
We shall now show how a simple refinement of the games model leads to a semantics for access control.

We shall assume as given a security semilattice $(\LL, {\sqcup}, {\bot})$, where $\sqcup$ is the binary join, and $\bot$ the least element. The partial order on the semilattice is defined by
\[ \ell \leq \ell' \; \equiv \; \ell \sqcup \ell' = \ell'. \]

We shall now form a category $\GL$, with a history-free sub-category $\GHL$, as follows. The objects of $\GL$ have the form
\[ A = (M_A,\lambda_A, \just_{A}, P_A,{\Deq}_A, \lev_{A}) \]
where $(M_A,\lambda_A, \just_{A}, P_A,{\Deq}_A)$ is a justified AJM game, and $\lev_{A} : M_{A} \rarr \LL$ assigns a security level to each move of the game. This additional piece of structure is carried through the type constructions in the simplest componentwise fashion:
\[ \lev_{A \otimes B} = \lev_{A \& B} = [\lev_{A}, \lev_{B}], \quad \lev_{\shr A} = [ \lev_{A} \mid i \in \omega] \]\[ \quad \lev_{A \linimpl B} = [ [\lev_{A} \mid b \in \Init_{B}], \lev_{B} ] . \]
The remainder of the definition of $\GL$  goes exactly as for $\GG$, with a single additional condition on plays in the definition of $\PP_{A}$:
\begin{description}
\item[(p6) Levels] A non-initial move $m$ can only be played if $\lev_{A}(m) \leq \lev_{A}(\just_{A}(m))$.
\end{description}
This constraint has a clear motivation, reflecting the basic intuition for access control: a principal can only affirm a proposition at its own level of authorization based on \emph{assertions made at the same level or higher}.
In terms of control flow (where the lattice has the opposite interpretation):  a procedure can only perform an action \emph{at its own security level or lower}.

Note that formally, this is a purely static constraint: it is used to discard certain moves (actions) at the level of the game (type), and independent of any particular play (run) or strategy (term).
This is remarkably simple, yet as we shall see, it suffices to soundly model the formal properties of the type theories which have been proposed for access control.

The content of this constraint is essentially the same as that described at a more concrete level in \cite{Mal99}.\footnote{This connection was pointed out to us by one of the referees.} What we have achieved here is to express this in a general, compositional form at the level of the semantic model. This allows general non-interference results to be proved, whereas the focus in \cite{Mal99} is on static analysis of specific programs.

\subsection{Level Monads}
For each AJM game $A$ and $\ell \in \LL$, there is a game $A_{\ell}$ in $\GL$ with $\lev_{A}(m) = \ell$ for all $m \in M_{A}$. Note that, fixing $\ell$, the assignment $A \mapsto A_{\ell}$ defines a full and faithful embedding of $\GG$ in $\GL$. The interesting structure of $\GL$ as a model for access control arises when there are moves at different levels.

We now define, for each $\ell \in \LL$, a construction $\Tl$ on games. This acts only on the level assignment, as follows:
\[ \lev_{\Tl A}(m) = \lev_{A}(m) \sqcup \ell . \]
All other components of $A$ remain unchanged in $\Tl A$. Note in particular that $P_{\Tl A} = P_{A}$. We must check that the Level condition \textbf{(p6)} is satisfied by plays  $s \in P_{A}$ with respect to $\lev_{\Tl A}$. This holds since $s$ satisfies \textbf{(p6)} with respect to $\lev_{A}$, and
\[ \ell_{1} \leq \ell_{2} \; \Rightarrow \; \ell_{1} \sqcup \ell \leq  \ell_{2} \sqcup \ell . \]

\noindent The following commutation properties of $\Tl$ are immediate.
\begin{proposition}
\label{commprop}
The following equations hold:
\[ \begin{array}{lcl}
\Tl I & = & I \\
\Tl (A \otimes B) & = & \Tl A \otimes \Tl B \\
\Tl (A \linimpl B) & = & \Tl A \linimpl \Tl B \\
\Tl (A \with B) & = & \Tl A \with \Tl B \\
\Tl \shr A & = & \shr \Tl A \\
\Tl (A \Rightarrow B)  & = &  \Tl A \Rightarrow \Tl B
\end{array}
\]
\end{proposition}

The semilattice structure on $\LL$ acts on the $\LL$-indexed family of monads in the evident fashion:
\begin{proposition}
\label{actprop}
The following equations hold:
\[ \begin{array}{lcl}
\Tl (T_{\ell'} A) & = & T_{\ell \sqcup \ell'} A \\
T_{\bot} A & = & A .
\end{array}
\]
\end{proposition}

\noindent We can extend each $\Tl$ with a functorial action: if $\sigma : A \rarr B$ then we can define $\Tl \sigma : \Tl A \rarr \Tl B$ simply by taking $\Tl \sigma = \sigma$. To justify this, note that
\[ P_{A \linimpl B} = P_{\Tl (A \linimpl B)} = P_{\Tl A \linimpl \Tl B} , \]
using Proposition~\ref{commprop}.
Hence $\sigma$ is a well-defined strategy for  $\Tl A \linimpl \Tl B$.

\begin{proposition}
\label{cclevprop}
The copy-cat strategy is well defined on $A \linimpl \Tl A$.
\end{proposition}
\begin{proof}
Consider a play of the copy-cat strategy
\[ \begin{array}{lccc}
& A & \linimpl & \Tl A \\
& \vdots & & \vdots \\
O & & & m_{1} \\
P & m_{1} & & \\
O & m_{2} & & \\
P & & & m_{2}
\end{array}
\]
which we write as $s = s_{1}m_{1}'m_{1}''m_{2}''m_{2}'$.
If $m_{1}$ is initial, $\lev_{A}(m_{1}) \leq \lev_{\Tl A}(m_{1})$, so the Level condition holds for $m_{1}''$.
If $m_{1}$ is non-initial, by Justification $m_{1}'$ is preceded by its justifier $m$ in $s \Rest \Tl A$. Since $s \Rest \Tl A \in P_{\Tl A} = P_{A}$, $\lev_{A}(m_{1}) \leq \lev_{A}(m)$, so $m_{1}''$ satisfies the Level condition in this case as well. Finally,
\[ \lev_{A}(m_{2}) \leq \lev_{A}(\just_{A}(m_{2})) \; \Rightarrow \; \lev_{\Tl A}(m_{2}) \leq \lev_{\Tl A}(\just_{A}(m_{2})) \]
so $m_{2}'$ satisfies the Level condition.
\end{proof}

\noindent Thus we can define a natural transformation
\[ \eta_{A} : A \rarr \Tl A \]
where $\eta_{A}$ is the copy-cat strategy.
Furthermore, by Proposition~\ref{actprop}, $\Tl \Tl A = \Tl A$.
Thus we obtain:
\begin{proposition}
Each $\Tl$ is an idempotent commutative monad.
\end{proposition}

\noindent A similar argument to that of Proposition~\ref{cclevprop} yields the following:
\begin{proposition}
\label{natprop}
If $\ell \leq \ell'$, then there is a natural transformation $\iota^{\ell,\ell'}_{A} : \Tl A \rarr T_{\ell'} A$, where each component is the copy-cat strategy.
\end{proposition}

\section{Non-Interference Results}

We now turn to the most important aspect of our semantics; the basis it provides for showing that certain kinds of data-access which would violate the constraints imposed by the security levels \emph{cannot in fact be performed}.

Firstly, we prove a strong form of converse of Proposition~\ref{natprop}.
\begin{proposition}
\label{nonatprop}
If $\neg(\ell \leq \ell')$, then there is no natural transformation from $\Tl$ to $T_{\ell'}$.
\end{proposition}
\begin{proof}
Suppose for a contradiction that there is such a natural transformation $\tau$.
Given any flat game $\XF_{\bot}$, with $\lev_{\XF_{\bot}}(m) = \bot$ for all moves $m \in M_{\XF_{\bot}}$,
the strategy $\tau_{\XF_{\bot}} : \Tl \XF_{\bot} \rarr \Tp \XF_{\bot}$ can only play in $\Tp \XF_{\bot}$, since playing the initial move in $\Tl \XF_{\bot}$ would violate the Level condition.

For readability,  in the remainder of the proof we let $A = \Nb$. Now consider the naturality square
\[ \begin{diagram}
\Tl A & \rTo^{\tau_{A}} & \Tp A \\
\dTo^{\Tl \sigma} & & \dTo_{\Tp \sigma} \\
\Tl A & \rTo_{\tau_{A}} & \Tp A \\
\end{diagram}
\]
Since $\tau_{A}$ can only play in $\Tp A$, for all $\sigma, \sigma' : A \rarr A$ we have $\Tl \sigma ; \tau_{A} = \Tl \sigma' ; \tau_{A}$, and hence $\tau_{A} ; \Tp \sigma = \tau_{A} ; \Tp \sigma'$ by naturality. Recall that $\Tp \sigma = \sigma$. But  we can take $\sigma = \{ \varepsilon, q_{0}0 \}$, $\sigma' = \{ \varepsilon, q_{0}1 \}$, and $q_{0}0 \in \tau_{A} ; \Tp \sigma \setminus \tau_{A} ; \Tp \sigma'$, yielding the required contradiction.
\end{proof}

The key step in the above argument was to show that  control  could not flow back from $\Tp \XF_{\bot}$ to the ``context'' $\Tl \XF_{\bot}$ because  its  security level $\ell$ is not below $\ell'$. We shall now extend this idea into an important general principle for the semantic analysis of access control.

\subsection{The No-Flow Theorem}

Consider the following situation. We have a term in context $\Gamma \vdash t:T$, and we wish to guarantee that $t$ is not able to access some part of the context. For example, we may have $\Gamma = x:U, \Gamma'$, and we may wish to verify that $t$ cannot access $x$. Rather than analyzing the particular term $t$, we may wish to guarantee this purely at the level of the types, in which case it is reasonable to assume that this should be determined by the types $U$ and $T$, and independent of $\Gamma'$.

This can be expressed in terms of the categorical semantics as follows. Note that the denotation of such a term in context will be a morphism of the form $f : A \otimes C \rarr B$, where $A = \lsem U \rsem$, $C = \lsem \Gamma'\rsem$, $B = \lsem T \rsem$.

\begin{definition}
Let $\CC$ be an affine category, \ie a symmetric monoidal category in which the tensor unit $I$ is the terminal object. We write $\top_{A} : A \rarr I$ for the unique arrow. We define $A \noflow B$ if for all objects $C$, and $f : A \otimes C \rarr B$, $f$ factors as
\[ f = A \otimes C \rTo^{\top_{A} \otimes \ident{C}} I \otimes C \rTo^{\cong} C \rTo^{g} B . \]
\end{definition}
The idea is that no information from $A$ can be used by $f$ --- it is ``constant in $A$''.
Note that $\GL$ and $\GHL$ are affine, so this definition applies directly to our situation.

Firstly, we characterize this notion in $\GL$ and $\GHL$.
\begin{lemma}
\label{noflemm}
In $\GL$ and $\GHL$, $A \noflow B$ if and only if, for any strategy $\sigma : A \otimes C \rarr B$, $\sigma$ does not play any move in $A$.
\end{lemma}
\begin{proof}
This reduces to verifying that $\sigma$ factors if and only if it plays no move in $A$. Certainly, if it factors it plays no move in $A$, since any such move in the composition must be preceded by one in $I$, which has none. Conversely, if it plays no move in $A$, then it is well-defined as a strategy $\sigma : C \rarr B$, and so it essentially factors through itself.
\end{proof}

We now give a simple characterization for when this ``no-flow'' relation holds between games.

Given a game $A$, we define:
\[ \begin{array}{lcl}
\Lev(A) & = & \{ \lev_{A}(m) \mid m \in \Init_{A} \} \\
A \rhd B & \equiv & \forall \ell \in \Lev(A), \ell' \in \Lev(B) . \, \neg(\ell \leq \ell')
\end{array}
\]

\begin{theorem}[No-Flow]
\label{nofth}
For any games $A$, $B$ in $\GL$:
\[ A \noflow B \;\; \Longleftrightarrow \;\; A \rhd B . \]
\end{theorem}
\begin{proof}
If $A \rhd B$, then any strategy $\sigma : A \otimes C \rarr B$ cannot play a move in $A$. The first such move would be an initial move in $A$, which would be justified by an initial move in $B$, and this would violate the Level condition since $A \rhd B$.

Conversely, suppose there are initial moves $m$ in $A$ and $m'$ in $B$ such that $\lev_{A}(m) \leq \lev_{B}(m')$. Then for any $C$, $\sigma = \{ \varepsilon, m'm \}$ is a strategy $\sigma : A \otimes C \rarr B$ which moves in $A$.
\end{proof}

\subsection{Computing Levels}

The characterization of no-flow in terms of the levels of types means that we can obtain useful information by computing levels.

We consider a syntax of types built from basic types (to be interpreted as flat games at a stipulated level) using the connectives of ILL extended with the level monads. For any such type $T$, we can give a simple inductive definition of $\Lev(A)$ where $A = \lsem T \rsem$:
\[ \begin{array}{lcl}
\Lev(\XF_{\ell}) & = & \{ \ell \} \\
\Lev(I) & = & \varnothing \\
\Lev(A \otimes B) & = & \Lev(A) \cup \Lev(B) \\
\Lev(A \linimpl B) & = & \Lev(B) \\
\Lev(A \with B) & = & \Lev(A) \cup \Lev(B) \\
\Lev(A \Rightarrow B) & = & \Lev(B) \\
\Lev(!A) & = & \Lev(A)\\
\Lev(\Tl A) & = & \{ \ell \sqcup \ell' \mid \ell' \in \Lev(A) \}
\end{array}
\]

\noindent This yields a simple, computable analysis which by Theorem~\ref{nofth} can be used to guarantee access constraints of the kind described above.

\subsection{Protected Types}
We give a semantic account of \emph{protected types}, which play a key r\^ole in the DCC type system \cite{AbadiBHR99}.

\begin{definition}
We say that a game $A$ is \emph{protected at level $\ell$} if $\Lev(A) \geq \ell$, meaning that $\ell' \geq \ell$ for all $\ell' \in \Lev(A)$.
\end{definition}
This notion extends immediately to types via their denotations as games.

The following (used as an inductive definition of protection in \cite{AbadiBHR99,Abadi06}) is an immediate consequence of the  definition.
\begin{lemma}
\begin{enumerate}
\item If $\ell \leq \ell'$, then $\Tp A$ is protected at level $\ell$.
\item If $B$ is protected at level $\ell$, so are $A \linimpl B$ and $A \Rightarrow B$.
\item If $A$ and $B$ are protected are level $\ell$, so are $A \with B$ and $A \otimes B$.
\item If $A$ is protected at level $\ell$, so is $\shr A$.
\item $I$ is protected at level $\ell$.
\end{enumerate}
\end{lemma}

We also have the following \emph{protected promotion} lemma, which shows the soundness of the key typing rule in DCC \cite{AbadiBHR99}.
\begin{lemma}
If $\sigma : \shr A \rarr \Tl B$, $\tau : \shr B \rarr  C$, and $C$ is protected at level $\ell$, then the coKleisli composition
\[ \sigma^{\dagger} ; \tau : \shr A \rarr C \]
is well-defined.
\end{lemma}
\begin{proof}
Firstly, by Proposition~\ref{commprop}, $\Tl \shr B = \shr \Tl B$.
So it suffices to show that $\tau$ is well-defined as a strategy $\tau : \Tl \shr B \rarr C$.
If we consider an initial move $m$ in $\Tl \shr B$ played by $\tau$, we must have $\lev_{\shr B}(m) \leq \lev(\just(m))$ since $\tau : \shr B \rarr C$ is well-defined. Moreover, $\ell \leq  \lev(\just(m))$ since $C$ is protected at $\ell$. Hence $\lev_{\Tl \shr B}(m) \leq \lev(\just(m))$.
\end{proof}

\subsection{Stability Under Erasure}

We now give a semantic version of the main result in \cite{Abadi06} (Theorem 7.6), which shows stability of the type theory under erasure of level constraints. This  is used in \cite{Abadi06} to derive several other results relating to non-interference.

Firstly, given $\ell \in \LL$, we define the erasure $A^{\ell}$ of a type $A$, which replaces every sub-expression of $A$ of the form $\Tp B$, with $\ell' \geq \ell$, by $\top$. Semantically, this corresponds to erasing all moves $m$ in the game (denoted by) $A$ such that $\lev(m) \geq \ell$, and all plays containing such moves.

Abadi's result is that, if we can derive a typed term in context $\Gamma \vdash e : A$, then we can derive a  term $\Gamma^{\ell} \vdash e' : A^{\ell}$.
To obtain an appropriate semantic version, we need to introduce the notion of \emph{total} strategies.
A strategy $\sigma$ is total if when $s \in \sigma$, and $sa \in P_{A}$, then $sab \in \sigma$ for some $b$. This is the direct analogue of totality for functions, and will hold for the strategies denoted by terms in a logical type theory --- although not in general for terms in a programming language equipped with general recursion. One can show that total strategies which are \emph{finite} (or alternatively \emph{winning}) in a suitable sense form a category with the appropriate structure to model intuitionistic and linear type theories \cite{Abr97,Hyl97}.

\begin{theorem}
\label{staberasth}
Suppose that $\sigma : A \rarr B$ is a total strategy. Then so is $\sigma' : A^{\ell} \rarr B^{\ell}$ for any
$\ell \in \LL$, where $\sigma'$ is the restriction of $\sigma$ to plays in $A^{\ell} \linimpl B^{\ell}$.
\end{theorem}
\begin{proof}
Suppose for a contradiction that $\sigma'$ is not total, and consider a  witness $sab \in \sigma \setminus \sigma'$, with $sa \in P_{A^{\ell} \linimpl B^{\ell}}$. Then $\lev(b) \geq \ell$; but by the Level constraint, we must have $\lev(\just(b)) \geq \ell$, which by the Justification condition contradicts $sa \in P_{A^{\ell} \linimpl B^{\ell}}$.
\end{proof}

\section{Semantic vs. Syntactic Non-Interference}
\label{semsynsec}

Our primary emphasis in this paper is on a semantic approach to access control, and we have proved semantic versions of a number of non-interference results. A detailed analysis of how these relate to the results proved by syntactic and proof-theoretic means for calculi such as DCC would take us too far afield. However, we shall provide an illustrative example of how semantic non-interference results can be used to obtain results about syntactic calculi.

For a simple and paradigmatic example, we consider a core fragment of DCC, obtained by extending the simply-typed $\lambda$-calculus with the level monads. There are two typing rules associated with the monads:
\[ \frac{\Gamma \vdash e : A}{\Gamma \vdash \eta_\ell e : A} \qquad \qquad
\frac{\Gamma \vdash e : T_{\ell} A \quad \Gamma , x : A \vdash e' : B}{\Gamma \vdash \bind : B}  \;\; \mbox{$B$ is protected at level $\ell$}
\]
The term rewriting rules, in addition to the usual $\beta$-reduction and $\eta$-expansion, are
\[ \eta_\ell e \longrightarrow e \qquad \qquad \bind \longrightarrow e'[e/x] \]
Thus the normal forms in this term calculus will be the usual long $\beta\eta$-normal forms of the simply typed $\lambda$-calculus. We say that a proof \emph{uses an assumption $x:A$} if the term corresponding to the proof contains $x$ free in its normal form.

It follows from the results in Section~3 that our game semantics provides a sound model for this calculus.
We have the following simple result.
\begin{proposition}
Let $\Gamma, x: A \vdash t:B$ be a term in context of core DCC, where $t$ is in long $\beta\eta$-normal form. Then $x$ occurs free in $t$ if and only if the strategy it denotes moves in $\lsem A\rsem$.
\end{proposition}
\begin{proof}
Given an occurrence of $x$ as a head variable in some sub-term of $t$, we can construct a play with appropriate choices of O-moves, which will ``activate'' this variable, whose denotation plays a copy-cat strategy with the occurrence of $x$ in the context, thus generating a move in $A$ as required.
The converse is easily shown by induction on normal forms.
\end{proof}
Combining this with Lemma~\ref{noflemm} and the No-Flow Theorem~\ref{nofth}, we immediately obtain:
\begin{proposition}
If $A \rhd B$, any derivation of $\Gamma, x: A \vdash t:B$ cannot use the assumption $x:A$.
\end{proposition}
Suitable adaptations of this argument to other type theories will yield corresponding non-interference results.

\section{Further Directions}

We have shown that Game Semantics provides a natural setting for the semantic analysis of access control. There are many further directions for this work:
\begin{itemize}
\item We have considered a semantic setting which is adequate for both intuitionistic and (intuitionistic-)linear type theories. It would also be interesting to look at access control in the context of \emph{classical type theories} such as $\lambda\mu$ \cite{Par92}, particularly since it is suggested in \cite{Abadi06,GP06} that there are problems with access control logics in classical settings. There have been some studies of game semantics for classical type theories \cite{Ong96,Lau04}. It would be of considerable interest to see if our approach could carry over to the classical case.
\item There are a number of other natural extensions, such as to polymorphic types.
\item It would also be of interest to develop the applications to dependency analysis for programming languages. The same game semantics framework provides a common basis for this and the study of logical type theories.
\item The development of algorithmic game semantics \cite{GM00,Abr2002a,MOW05,MW05}, including several implemented verification tools \cite{BG07,DL07,LMOW}, suggests that it may be promising to look at automated analysis based on our semantic approach.
\item We have developed our semantics in the setting of AJM games, equipped with a notion of justification. One could alternatively take HO-games as the starting point, but these would also have to be used in a hybridized form, with ``AJM-like'' features, in order to provide models for linear type theories \cite{McC98}. In fact, one would like a form of game semantics which combined the best features (and minimized the disadvantages) of the two approaches. Some of the ideas introduced in the present paper may be useful steps in this direction.
\end{itemize}

\bibliographystyle{plain}

\end{document}